\documentclass[11pt]{article}
\usepackage{times}

\newtheorem{definition}{Definition}
\newtheorem{rem}[definition]{Remark}

\newtheorem{lemma}[definition]{Lemma}
\newtheorem{theorem}[definition]{Theorem}
\newtheorem{proposition}[definition]{Proposition}

\usepackage{verbatim}
\usepackage{graphicx}
\usepackage{amssymb}
\usepackage{amsfonts,amsmath}
\usepackage[numbers]{natbib}
\usepackage{amsopn}
\usepackage{amssymb}
\usepackage{caption}
\usepackage{subcaption}
\usepackage{accents}
\usepackage{epsfig}
\usepackage{epstopdf}
\usepackage{algorithm}
\usepackage{algpseudocode}
\usepackage{algorithmicx}
\usepackage{changepage}
\usepackage{color}

\def\RR{\hbox{\sf I\kern-.14em\hbox{R}}}

\newcommand{\ind}{\mathrm{IND}}
\newcommand{\weight}{\mathrm{WEIGHT}}
\newcommand{\father}{\mathrm{FATHER}}
\newcommand{\successor}{\mathrm{ORIEN}}

\newcommand{\dfn}{\mathrm{DFN}}
\newcommand{\rooot}{\mathrm{ROOT}}
\newcommand{\mmark}{\mathrm{MARK}}
\newcommand{\last}{\mathrm{Last}}

\begin{document}
\global\def\refname{{\small \bf References}}
%
%

\centerline{\Large \bf
       Weighted domination number of cactus graphs
}
\date{}
\vspace{7mm}
\centerline{\small{\bf Tina Novak}}
{\it
\centerline{\small University of Ljubljana, Faculty of Mechanical Engineering}
\centerline{\small A\v sker\v ceva 6, SI-1000 Ljubljana, Slovenia}
\centerline{\small tina.novak@fs.uni-lj.si }
}
\vspace{7mm}
\centerline{\small{\bf Janez \v{Z}erovnik}}
{\it
\centerline{\small University of Ljubljana, Faculty of Mechanical Engineering}
\centerline{\small A\v sker\v ceva 6, SI-1000 Ljubljana, Slovenia}
\centerline{\small janez.zerovnik@fs.uni-lj.si }
}
\vspace{7mm}

%

%
%
%
%

\noindent{\small {\bf Abstract:} In the paper, we write a linear algorithm for calculating the weighted domination number of a vertex-weighted cactus. The algorithm is based on the well known depth first search (DFS) structure. Our algorithm needs less than $12n+5b$ additions and $9n+2b$ $\min$-operations where $n$ is the number of vertices and $b$ is the number of blocks in the cactus.}

\vspace{7mm}

\noindent{\small {\bf Keywords:} weighted domination problem,cactus graph, DFS structure

}

\section{Introduction}

\label{section:intro}

Cactus graphs are interesting generalizations of trees,
with numerous applications, for example  in location theory \cite{Boaz,ZZTOCP},
communication networks \cite{Fink,Blaz}, stability analysis \cite{Arcak}, and elsewhere.
 Usually, linear problems on trees imply linear problems on cacti. In this paper, we study the weighted domination number of a cactus graph with weighted vertices. It is well known that the problem of the weighted domination number on trees is linear \cite{CAAODIG,NWODIWT}. Actually, we also have very general linear algorithm for computing domination-like problems on partial $k$-trees \cite{PAPkT}. Time complexity of this algorithm is ${\cal O}(n|L|^{2k+1})$, where $k$ is the treewidth and $L$ is the set of vertex states (the different ways that a solution to a subproblem impact to the origin vertex). In the case of cactus graphs we have $k=2$ and $|L|=3$. Therefore, the time complexity of the general algorithm  \cite{PAPkT} on cacti is ${\cal O}(3^{5}n)$.

It is well known that cactus graphs can be recognized by running an extended version of  depth first search (DFS) algorithm that results a data structure of a cactus, see for example \cite{ZZCWW}.
From the data structure, the  vertices can be naturally divided  into three types, i.e. each vertex either lies on a cycle and has degree $2$ or lies on a cycle and has degree $\geq 3$ or does not lie on a cycle (see \cite{BKLAPN}).
Using this structure, we design an algorithm for general cacti.
In the paper, we first illustrate the basic idea by writing a version of the algorithm for trees
before  generalizing  the approach to arbitrary  cactus graphs.
%
%
Our   algorithm has time complexity ${\cal O}(28n)$ which substantially improves the constant $3^5 =243$.
In fact, we  will estimate time complexity of our algorithm  more precisely (blocks will be formally  defined later)

\begin{theorem}
   Let $n$ be the number of vertices in a cactus and $ b < n $ be the number of blocks.
   For computing the  weighted domination number we need less than $12n+5b$ additions and $9n+2b$ $\min$-operations.
\end{theorem}

The rest of the paper is organized as follows.
In  the next section we first recall   definitions of the domination number and the weighted domination number of general graphs. For cacti, we introduce the classification of vertices in relation to the skeleton structure \cite{BKLAPN}.
In Section \ref{section:WDP} we define three parameters that are useful when considering  the weighted domination problem. The  simplified version of the algorithm that is used for computing   the weighted domination number of a tree is presented in Section \ref{section:tree}. Special cases of graphs, i.e. path-like graphs and cycle-like graphs are regarded  in Section \ref{section:subalg}. We write algorithms for calculating their weighted domination parameters and weighted domination number.
In Section \ref{section:summary}, the algorithm for general cacti is given and its time complexity is estimated.


\section{Definitions and preliminaries}

\label{section:def}

\noindent
{\bf A vertex-weighted graph and the weighted domination number}

Let $G=(V,E)$ be a graph with a set of vertices $V=V(G)$ and a set of edges $E=E(G)$. Denote by $N(v)$ the open neighborhood of a vertex $v$ i.e. the set of vertices adjacent to the vertex $v$ and by $N[v]$ the closed neighborhood of a vertex $v$: $N[v]=\{v\} \cup N(v)$. Let $S$ be any subset of the set of vertices $V$. Denote by $N(S)$ the open neighborhood of the set $S$ i.e. the set of vertices adjacent to any vertex in $S$ and similarly by $N[S]$ the closed neighborhood of $S$: $N[S]=S \cup N(S)$.
A subset $D \subseteq V$ is a dominating set if $N[D]=V$. A domination number $\gamma(G)$ is the minimum cardinality among all dominating sets of the graph $G$.

In this article, a weighted graph $(G,w)$ is a graph together with a positive real weight-function $w:V \to {\mathbb R}_{+}$. For the vertex $v_{j} \in V$ we shall write $w_{j}=w(v_{j})$. The weight of a dominating set $D$ is defined as $w(D)=\sum_{v_{j} \in D} w_{j}$. Finally, the weighted domination number (WDN) $\gamma_{w}(G)$ of the graph $G$ is the minimum weight of a dominating set, more precisely
\begin{equation}
   \gamma_{w}(G) = \min\Big\{ w(D)\, \big| \; D \; \hbox{is a dominating set} \Big\}  \;.
   \label{eq:WeightedDominationNumber}
\end{equation}


\noindent
{\bf Cactus graph and its skeleton}

A graph $K=(V(K),E(K))$ is a cactus graph if and only if any two cycles of $K$ have at most one vertex in common.
 Equivalently, any edge of a cactus lies on at most one cycle.
Skeleton structure of a cactus is elaborated  in \cite{BKLAPN}, where it is
shown that the vertices of a cactus graph are of three types:
\begin{enumerate}
   \item [$\bullet$]  $C$-vertex is a vertex on a cycle of degree 2,
   \item [$\bullet$]  $G$-vertex is a vertex not included in any cycle,
   \item [$\bullet$]  $H$-vertex or a hinge is a vertex which is included in at least one cycle and is of degree $\geq 3$.
\end{enumerate}
By a {\em subtree} in a cactus we mean a tree induced by a subset of $G$-vertices and $H$-vertices only.
 A {\em graft} is a maximal subtree in a cactus.
 A subgraph of a cactus is called  a {\em block} when it is either a cycle or a graft.
\\

\noindent
{\bf Depth First Search  (DFS) algorithm}

The DFS is a well known method for exploring graphs.
It can be used for recognizing cactus graphs providing  the data structure (see \cite{GTAA}, \cite{ZZCWW}, \cite{ZZACP}, \cite{ZZTOCP}).
Consider, we have a cactus graph $K$. We can distinguish one vertex as a root of $K$ and denote it by $r$.
After running the  DFS algorithm, the vertices of $K$ are DFS ordered. The order is given by the order in which DFS visits the vertices. (Note that the DFS order of a graph is not unique as we can use any vertex as the starting vertex (the root)  and can visit the neighbors of a vertex in any order. However, here we can assume that the DFS order is given and is fixed.)

We denote by  $\dfn(v)$  the position of $v$ in the DFS order and we set $\dfn(r)=0$.
$\dfn$ is called the depth first number. Following \cite{ZZCWW} and \cite{ZZACP},
it is useful to store the information recorded during the DFS run in four arrays, called the
DFS (cactus) data structure:
\begin{enumerate}
   \item [$\bullet$] $\father(v)$ is the unique predecessor (father) of vertex $v$ in the rooted tree, constructed with the DFS.
   \item [$\bullet$] $\rooot(v)$ is the root vertex of the cycle containing $v$ i.e. the first vertex of the cycle (containing $v$) in the DFS order. If $v$ does not lie on a cycle, then $\rooot(v)=v$. We set $\rooot(r)=r$. (In any DFS order, if $\dfn(w) < \dfn(v)$ and $w$ is the root of the cycle containing $v$ and $v$ is the root of another cycle (it is a hinge), then $\rooot(v)=w$.)
   \item  [$\bullet$] For vertices on a cycle (i.e. $\rooot(v) \not= v$), orientation of the cycle is given by
        $\successor(v) =z$, where $z$ is the son of   $\rooot(v)$ that is visited on the cycle first.
        If $\rooot(v) = v$ , then $\successor(v)=v$.
   \item [$\bullet$] $\ind(v) := |\{u\,\big|\; \father(u)=v\}|$ is the number of sons of $v$ in the DFS tree.
\end{enumerate}

Below we write the pseudocode of the DFS algorithm that provides the data structure of cacti. The idea is taken from \cite{GTAA}.
    To mark a visited vertex in the procedure, we introduce auxiliary array $\mmark$ (as in \cite{GTAA}). At the beginning of the algorithm, we set MARK$(v)=0$ for every vertex in $K$. During the algorithm, whenever a vertex $v$ is visited for the first time, the value MARK$(v)$ becomes $1$ and DFN$(v)$ is increased by $1$.\\
\begin{algorithm}{}
 \caption{DFS algorithm}
 \begin{algorithmic}
 \State
 {\bf Data:} {Rooted cactus $(K,r)$ with vertices $V(K)$ and edges $E(K)$;} \\
 {\bf initialize}\\
       \quad $i=0$; \\
       \quad For every vertex $v$ in $K$ set \\
       \qquad $\father(v)=v$;
       \quad $\mmark(v)=0$;
       \quad $\rooot(v)=v$; \\
       \qquad $\successor(v)=v$;
       \quad $\ind(v)=0$;
       \quad $\dfn(v)=0$; \\
       \quad and for the root $r$ reset: \\
       \qquad $\mmark(r)=1$; \\
       \quad $v=r$;
    \end{algorithmic}
    \label{algorithm:dfs:skeleton:1}
\end{algorithm}

\begin{algorithm}
   \begin{algorithmic}
   \State
   {\bf repeat}
      \begin{algorithmic}
          \If{{$\hbox{\rm all the edges incident to}  \; v \; \hbox{\rm have already been labeled
                "examined"} $}
                \State ($v$  {\it is completely scaned})}
                \State $v = \father(v)$
          \Else {\it \ (an edge $(v,w)$ is not labeled "examined")}
                 \State The edge $(v,w)$ label "examined" and {\bf do} the following
                  \If{\rm $\mmark(w)=0$}
                     \State $i=i+1$;
                     \State $\dfn(w)=i$;
                     \State $\mmark(w)=1$;
                     \State $\father(w)=v$;
                     \State $\ind(w)=\ind(v)+1$;
                     \State $v=w$.
                  \Else {\em ($\mmark(w)=1$, that means we have a cycle)}
                     \State label the edge $(w,v)$ "examined";
                     \State $\rooot(v)=w$;
                     \State $u=\father(v)$;
                     \State {\bf repeat} {\it (assigning the root $w$ of vertices of the cycle)}
                              \State\qquad  $z=u$;
                              \State\qquad  $\rooot(z)=w$;
                              \State\qquad  $u=\father(z)$;
                     \State {\bf until} $u=w$. {\it (now $z$ determines the orientation of the cycle with the
                              \State \qquad \qquad \quad \;\; root $w$)}
                     \State {\bf repeat} {\it (assigning the successor $z$ i.e. the orientation of vertices of 
                     \State \qquad \quad \; the cycle)}
                              \State\qquad $\successor(v)=z$;
                              \State\qquad $v=\father(v)$;
                     \State {\bf until} $v=w$.
                     \State $v=w$;
                  \EndIf
       \EndIf
       \end{algorithmic}
    \end{algorithmic}
    {\bf until} $v=r$ and all edges incident to $r$ are "examined" \\
    {\bf Result:} arrays $\father$, $\rooot$, $\successor$, $\ind$, $\mmark$, $\dfn$.
 \medskip
 \caption{DFS algorithm - Part 2}
 \label{algorithm:dfs:skeleton:2}
\end{algorithm}

\noindent
Direct  correspondence of the definitions of $C$, $G$, $H$-vertices in a rooted cactus $(K,r)$ and arrays $\father$, $\rooot$, $\successor$ and $\ind$ is described in the following lemma
\begin{lemma}[($C$,$G$,$H$-vertices in DFS array)] $  $
   \begin{enumerate}
      \item For a vertex $v \ne r$ the following holds
            \begin{enumerate}
               \item $v$ is a $C$-vertex if and only if $\rooot(v) \ne v$ and $\ind(v)=1$
               \item $v$ is a $G$-vertex if and only if $\rooot(v)=v$ and $\successor(v)=v$ and for every son $u$ of $v$ we have $\rooot(u) \ne v$
               \item $v$ is a $H$-vertex if and only if either ($\rooot(v)=v$ and $\successor(v) = v$ and for at least one son $u$ of $v$ we have $\rooot(u) =v$ ) or ($\rooot(v) \ne v$ and $\ind(v)>1$).
            \end{enumerate}
      \item For the root $r$ we have
            \begin{enumerate}
               \item $r$ is a $C$-vertex if and only if $\ind(r) = 1$ and for the son $u$ of $r$ ($\dfn(u)=1$) we have $\rooot(u)=v$
               \item $r$ is a $G$-vertex if and only if for every son $u$ of $r$ we have $\rooot(u)=u$
               \item $r$ is a $H$-vertex if and only if $\ind(r) > 1$ and for at least one son $u$ of $r$ we have $\rooot(u)=r$.
            \end{enumerate}
   \end{enumerate}
\end{lemma}

\begin{rem}
   For any vertex $v \in V(K)$ and his father $w=\father(v)$, vertices with $\dfn$'s
   $$   \dfn(w),\dfn(w)+1,\ldots,\dfn(v)-1 $$
   (and all corresponding edges induced by  $V(K)$) form a  rooted subcactus  with the root $w$, denote it
$(\widetilde{K}_{w},w)$. Graphs $\widetilde{K}_{w}$ and $\{v\}$ are disjoint.
   \label{remark:subcactus}
\end{rem}

\noindent
{\bf Observation.}
 Assume
 the last vertex $l$ in the DFS order of a cactus $K$ lies on a subtree $T$ in $K$. Let $w = \father(l)$ and $w$ be the root (according to DFS order) of any subcactus $\widetilde{K}_{w}$, such that $\{l\} \cap V(\widetilde{K}_{w})=\emptyset$. If $\widetilde{v} \in  V(\widetilde{K}_{w})$, then $\dfn(w) \leq \dfn(\widetilde{v}) < \dfn(l)$.
Similar but perhaps a little less obvious fact  is given in  the next proposition.

\begin{proposition}
   Consider the last vertex $l$ in DFS order of a cactus $K$ lies on a cycle $C$. Then the following is true
   \begin{enumerate}
      \item  the  neighboring vertex of $l$ in the cycle $C$, which is not the father of the vertex $l$,
it is the root of the cycle $C$.
      \item   vertex $l$ is not a hinge.
      \item let  $w,v \in C$, $w=\father(v)$, $w$ is not the root of the cycle $C$ and $w$ is a hinge, i.e. the root of a subcactus $\widetilde{K}$, such that $V(C) \cap V(\widetilde{K}) = w$.
          For any  $\widetilde{v} \in V(\widetilde{K})$, we have $\dfn(w)\leq\dfn(\widetilde{v})<\dfn(v)$.
   \end{enumerate}
   \label{prop:last vertex cycle}
\end{proposition}

\begin{proof}
   \begin{enumerate}
      \item  Denote by $v$ a neighboring vertex of $l$ in the cycle $C$, which is not the father of $l$. If $v$ is not the root of $C$, then $\dfn(v) > \dfn(l)$. Contradiction.
      \item  If $l$ is a hinge, according to DFS order, there exist at least one vertex with $\dfn > \dfn(l)$. Contradiction.
      \item  According to DFS order, the inequality  $\dfn(w)\leq\dfn(\widetilde{v})$  holds. Consider there is  $\widetilde{v} \in \widetilde{K}$ with $\dfn(\widetilde{v}) > \dfn(v)$. Following DFS algorithm, we have then $\dfn(\widetilde{v}) > \dfn(l)$. Contradiction.
   \end{enumerate}
\end{proof}

\section{Weighted domination parameters (WDP)}

\label{section:WDP}

Let $G$ be a graph and $v$ any vertex in $V(G)$.
Consider the following three parameters yielding related weighted domination parameters (see \cite{CAAODIG}):
\begin{definition} $ $
   \begin{enumerate}
      \item $\gamma_{w}^{00}(G,v)=\min\big\{w(D)\,\big| \; D \; \hbox{is a dominating set of}\;\,G-v\big\}=\gamma_{w}(G-v)$
      \item $\gamma_{w}^{1}(G,v)=\min\big\{w(D)\,\big| \; D \; \hbox{is a dominating set of}\;\,G\;\hbox{and}\;v \in D\big\}$
      \item $\gamma_{w}^{0}(G,v)=\min\big\{w(D)\,\big| \; D \; \hbox{is a dominating set of}\;\,G\;\hbox{and}\;v \notin D\big\}$.
   \end{enumerate}
\end{definition}

It is obvious that
\begin{equation}
   \gamma_{w}(G) = \min\big\{ \gamma_{w}^{1}(G,v), \gamma_{w}^{0}(G,v)\big\} \,.
   \label{eq:WDNGu}
\end{equation}
Since a dominating set of $G$, which does not contain the vertex $v$ is also a dominating set of $G-v$, we have the relation
\begin{equation}
   \gamma_{w}^{00}(G,v) \leq \gamma_{w}^{0}(G,v) \,.
   \label{rel:00and0}
\end{equation}
Let $D$ be a dominating set of $G-v$ such that $w(D)=\gamma_{w}(G-v)$. Then $D \cup \{v\}$ is a dominating set of $G$ and clearly
\begin{equation}
   \gamma_{w}^{1}(G,v) \leq w(v) +\gamma_{w}^{00}(G,v) \,.
   \label{rel:1and00}
\end{equation}

\begin{lemma}
   Let  $G_{1}$ and $G_{2}$ be disjoined rooted graphs with roots $v_{1}$ and $v_{2}$ respectively, and let $G$ be a disjoint union of $G_{1}$ and $G_{2}$ joined by the edge $v_{1}v_{2}$. Then the following is true:
   \begin{enumerate}
      \item $\gamma_{w}^{00}(G,v_{1}) = \gamma_{w}^{00}(G_{1},v_{1})+\gamma_{w}(G_{2})$,
      \item $\gamma_{w}^{1}(G,v_{1}) = \gamma_{w}^{1}(G_{1},v_{1})+\min\big\{\gamma_{w}^{1}(G_{2},v_{2}), \gamma_{w}^{00}(G_{2},v_{2})\big\}$,
      \item $\gamma_{w}^{0}(G,v_{1}) = \min\big\{\gamma_{w}^{0}(G_{1},v_{1})+\gamma_{w}(G_{2},v_{2}), \gamma_{w}^{00}(G_{1},v_{1})+\gamma_{w}^{1}(G_{2},v_{2})\big\}$.
   \end{enumerate}
   \label{lemma:WDNofGRAPH2vertices}
\end{lemma}

The proof of Lemma \ref{lemma:WDNofGRAPH2vertices} (for the domination number) appears in \cite{CAAODIG}.
Generalization to weighted domination is straightforward and therefore ommited.
 A more general situation is  described by the next lemma

\begin{lemma}
   Let $G_{1}$ and $G_{2}$ be graphs with one common vertex $v_{0}$ and let $G_{1}-v_{0}$ and $G_{2}-v_{0}$ be disjoined.  Denote by $G$ the union of $G_{1}$ and $G_{2}$. Then we have
   \begin{enumerate}
      \item $\gamma_{w}^{00}(G,v_{0}) = \gamma_{w}^{00}(G_{1},v_{0})+\gamma_{w}^{00}(G_{2},v_{0})$,
      \item $\gamma_{w}^{1}(G,v_{0}) = \gamma_{w}^{1}(G_{1},v_{0}) + \gamma_{w}^{1}(G_{2},v_{0})-w(v_{0})$,
      \item $\gamma_{w}^{0}(G,v_{0}) = \min\big\{\gamma_{w}^{0}(G_{1},v_{0})+\gamma_{w}^{00}(G_{2},v_{0}),\gamma_{w}^{00}(G_{1},v_{0})+\gamma_{w}^{0}(G_{2},v_{0})\big\}$\,.
   \end{enumerate}
   \label{lemma:WDNofGRAPH1vertex}
\end{lemma}

\begin{proof}
   \begin{enumerate}
      \item
As  $G_{1}-v_{0}$ and $G_{2}-v_{0}$ are disjoined, it follows
$\gamma_{w}^{00}(G,v_{0}) = \gamma_{w}^{00}(G_{1},v_{0})+\gamma_{w}^{00}(G_{2},v_{0})  .$
       \item
Let $D$ be a dominating set of $G$ with $v_0 \in D$ such that  $w(D) = \gamma_{w}^{1}(G,v_{0})$.
Then
$D_1 = D \cap V(G_1)$ is a dominating set of  $ G_1$ and  $w(D_1) \geq  \ \gamma_{w}^{1}(G_{1},v_{0}) $.
Similarly,
 $D_2 = D \cap V(G_2)$ is a dominating set of  $ G_2$ and  $w(D_2) \geq  \gamma_{w}^{1}(G_{2},v_{0}) $.
Hence   $\gamma_{w}^{1}(G,v_{0}) \geq \gamma_{w}^{1}(G_{1},v_{0}) + \gamma_{w}^{1}(G_{2},v_{0})-w(v_{0})$.
On the other hand, for any dominating sets $D_1$ and $D_2$ with
 $w(D_1) \geq  \ \gamma_{w}^{1}(G_{1},v_{0}) $ and     $w(D_2) \geq  \gamma_{w}^{1}(G_{2},v_{0}) $,
$D = D_1 \cup D_2$ dominates $G$.
As  $D_1 \cap D_2 = \{v_0\}$,  we have $w(D) = w(D_1) +w(D_2) -  w(v_0)$     and therefore
 $\gamma_{w}^{1}(G,v_{0}) \leq  w(D)  = \gamma_{w}^{1}(G_{1},v_{0}) + \gamma_{w}^{1}(G_{2},v_{0}) - w(v_{0})$.
       \item
As we consider only dominating sets with $v_0 \not\in D$,
$v_0$ has to be dominated by some other vertex.
We distinguish three cases: either $v_0$ is dominated by  $D_1 = D \cap V(G_1)$ ,    or  $D_2 =D  \cap V(G_2)$, or by
both $D_1$ and $D_2$.
Assuming  $w(D) =  \gamma_{w}^{0}(G,v_{0})$, and
recalling   that   $\gamma_{w}^{00}(G_{i},v_{0}) \leq \gamma_{w}^{0}(G_{i},v_{0}$ for $i=1,2$,
it follows
\begin{eqnarray}
  \gamma_{w}^{0}(G,v_{0})  &\!\!= &
 \!\!\min\big\{\gamma_{w}^{0}(G_{1},v_{0})+\gamma_{w}^{00}(G_{2},v_{0}),\gamma_{w}^{00}(G_{1},v_{0})+\gamma_{w}^{0}(G_{2},v_{0}),    \nonumber\\
  &   &  \!\!\qquad  \;\, \gamma_{w}^{0}(G_{1},v_{0})+\gamma_{w}^{0}(G_{2},v_{0}) \big\} \geq
   \nonumber\\
 & \!\! \geq & \!\!\min\big\{\gamma_{w}^{0}(G_{1},v_{0})+\gamma_{w}^{00}(G_{2},v_{0}),\gamma_{w}^{00}(G_{1},v_{0})+\gamma_{w}^{0}(G_{2},v_{0})\big\} \,.  \nonumber
\end{eqnarray}
On the other hand,   we can
construct dominating sets of $G$ by taking a union of two dominating sets $D_1$ and $D_2$ of $G_1$ and $G_2$ respectively. At least one of $D_1$, $D_2$  (or both) must dominate $v_0$.
Taking either ($w(D_1) = \gamma_{w}^{0}(G_{1},v_{0})$ and $w(D_2) = \gamma_{w}^{00}(G_{2},v_{0})$)
or  ($w(D_1) = \gamma_{w}^{00}(G_{1},v_{0})$ and $w(D_2) = \gamma_{w}^{0}(G_{2},v_{0})$),
we conclude
\begin{eqnarray}
 \min\big\{\gamma_{w}^{0}(G_{1},v_{0})+\gamma_{w}^{00}(G_{2},v_{0}),\gamma_{w}^{00}(G_{1},v_{0})+\gamma_{w}^{0}(G_{2},v_{0}),    \nonumber\\
  \gamma_{w}^{0}(G_{1},v_{0})+\gamma_{w}^{0}(G_{2},v_{0}) \big\}   &\!\! \geq &
  \nonumber  \\
\geq \min\big\{\gamma_{w}^{0}(G_{1},v_{0})+\gamma_{w}^{00}(G_{2},v_{0}),\gamma_{w}^{00}(G_{1},v_{0})+\gamma_{w}^{0}(G_{2},v_{0})\big\}  & \!\!\geq &   \!\!  \gamma_{w}^{0}(G,v_{0}) .  \nonumber
\end{eqnarray}
\end{enumerate}
\end{proof}

\section{Algorithm for trees }

\label{section:tree}

In this section, let  $G=(V,E)$ be  a vertex-weighted tree,
and let $T$ be an associated rooted tree with root $r$
($r$ can be  arbitrary but fixed vertex in  $V(G)$
).
In \cite{NWODIWT}, the authors write the algorithm for calculating the weighted domination of a vertex-edge-weighted tree. It can of course be applied to a vertex-weighted tree, the case of interest in this paper.
Another algorithm for calculating the weighted domination number of a tree appears in \cite{CAAODIG}.
We write a new  algorithm for weighted domination number of a  weighted tree
based on the DFS data structure  here in order to illustrate the main idea on a well understood special case
in order to clarify the development of the general algorithm in the following sections.

\begin{figure}[h]
   \centering
      \includegraphics[width=0.27\textwidth]{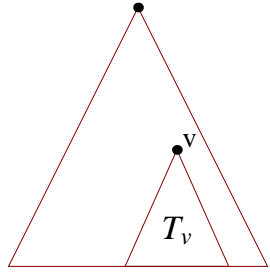}\\
      \caption{{\small The rooted subtree}}
      \label{fig:RootedSubtree}
\end{figure}

Denote by $(T_{v},v)$ the rooted subtree with the root $v$ as is shown in Figure \ref{fig:RootedSubtree}.
In our algorithm we use the DFS order of vertices (i.e. the   DFS cactus data structure provided by the DFS algorithm).
We supplement the DFS data structure by four arrays of the initial values of the parameters $\gamma_{w}^{00}$, $\gamma_{w}^{1}$, $\gamma_{w}^{0}$ and $\gamma_{w}$.
  Initially,   we set for every vertex $v$
\begin{equation}
   \gamma_{w}^{00}(v)=0, \quad \gamma_{w}^{1}=w(v), \quad \gamma_{w}^{0}=\infty \quad \hbox{and} \quad \gamma_{w}(v)=w(v) \,.
\end{equation}
The algorithm's starting point is the last vertex $v$ in the DFS order with the  corresponding parameters $\gamma_{w}^{00}(v)$, $\gamma_{w}^{1}(v)$, $\gamma_{w}^{0}(v)$ and $\gamma_{w}(v)$. In the
data structure  we find the father of $v$ and  call  it $w$.
If $\dfn(w) \ne \dfn(v)-1$ (i.e. $\dfn(w) < \dfn(v)-1$), there exists rooted subtree $(\widetilde{T}_{w},w)$ (see Remark \ref{remark:subcactus}). The algorithm calls itself  recursively for the subtree $\widetilde{T}_{w}$ and
then accordingly updates
the parameters $\gamma_{w}^{00}(w)=\gamma_{w}^{00}(\widetilde{T}_{w},w)$, $\gamma_{w}^{1}(w)=\gamma_{w}^{1}(\widetilde{T}_{w},w)$, $\gamma_{w}^{0}(w)=\gamma_{w}^{0}(\widetilde{T}_{w},w)$ and $\gamma_{w}(w)=\gamma_{w}(\widetilde{T}_{w})$.
When $w$ and $v$ are the last two vertices in the DFS order, the parameters at $w$
are computed according to
 Lemma \ref{lemma:WDNofGRAPH2vertices}, and the computation continues regarding $w$ as the last vertex.
For pseudocode of the algorithm see Algorithm \ref{algorithm:RootedTree}.

\begin{figure}[h]
   \centering
      \includegraphics[width=0.27\textwidth]{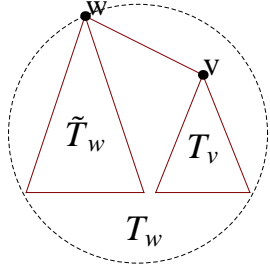}\\
      \caption{{\small Subtrees $T_{v}$, $\widetilde{T}_{w}$ and $T_{w}$}}
      \label{fig:Treesvw}
\end{figure}

\begin{algorithm}[H]
     \caption{TREE}
    \begin{algorithmic}
    \State
    {\bf Data:} A rooted tree $(T,r)$ with DFS ordered vertices in the DFS table \\
    {\bf initialize} $\gamma_{w}^{00}(v)=0$, $\gamma_{w}^{1}(v)=w(v)$, $\gamma_{w}^{0}(v)=\infty$ and $\gamma_{w}(v)=w(v)$ for
                     every vertex $v$ in the DFS table  \\
    {\bf set} $v$ is the last vertex in the DFS order; \\
    {\bf repeat}
       \begin{algorithmic}
           \State $w=\father(v)$;
           \If{$\dfn(w) \ne \dfn(v)-1$}
                          \State {\bf call} algorithm {\bf TREE} for the rooted tree on vertices with \\
                          \quad \; $\dfn = \dfn(w),\ldots,\dfn(v)-1$ and the root $w$
                          {\it (we obtain new values \\
                          \quad \; for $\gamma_{w}^{00}(w)$, $\gamma_{w}^{1}(w)$, $\gamma_{w}^{0}(w)$ and $\gamma_{w}(w)$)};
           \EndIf \\
           $\gamma_{w}^{00}(w) = \gamma_{w}^{00}(w)+\gamma_{w}(v)$; \\
           $\gamma_{w}^{1}(w) = \gamma_{w}^{1}(w) + \min\{\gamma_{w}^{1}(v),\gamma_{w}^{00}(v)\}$; \\
           $\gamma_{w}^{0}(w) = \min\{ \gamma_{w}^{0}(w)+\gamma_{w}(v), \gamma_{w}^{00}(w)+\gamma_{w}^{1}(v)\}$; \\
           $\gamma_{w}(w) = \min\{\gamma_{w}^{1}(w), \gamma_{w}^{0}(w) \}$; \\
           $v=w$;
       \end{algorithmic}
    {\bf until} $v=r$ \\
    {\bf Result:} $\gamma_{w}^{*}(T,r)=\gamma_{w}^{*}(v)$ \, for \, $*=00,1,0$;
     \begin{algorithmic}
      \State \qquad \quad $\gamma_{w}(T) = \gamma_{w}(v)$.
     \end{algorithmic}
     \end{algorithmic}
 \label{algorithm:RootedTree}
\end{algorithm}

\newpage
\begin{proposition}[(Time complexity of TREE)]
   Algorithm TREE needs $4(n-1)$ additions and $3(n-1)$ $\min$-operations.
   \label{prop:TimeComplexityTREE}
\end{proposition}
\begin{proof}
   Using Lemma \ref{lemma:WDNofGRAPH2vertices} and the equation $\gamma_{w}(T_{w}) = \min\{\gamma_{w}^{1}(T_{w},w),\gamma_{w}^{0}(T_{w},w) \}$ in a step of the algorithm for rooted subtrees $(T_{w},w)$ and $(T_{v},v)$ (where $w=\father(v)$), the calculation demands $4$ additions and $3$ $\min$-operations.
   The algorithm sticks rooted subtrees $(T_{w},w)$ and $(T_{v},v)$ for every existing edge $(w,v)$.
\end{proof}

\section{Cacti - more lemmas and subalgorithms}

\label{section:subalg}

The   algorithm for cactus graph should exploit the tree structure obtained from DFS representation. It would be meaningful to preserve the form of algorithm TREE if the current vertex of a cactus lies on a tree. Special attention should be paid to the current vertex on a cycle. In this section we prepare subalgorithm CYCLE-LIKE for the rooted cycle $(C,r)$, which calculates parameters $\gamma_{w}^{00}(C,r)$, $\gamma_{w}^{1}(C,r)$, $\gamma_{w}^{0}(C,r)$ and $\gamma_{w}(C)$.


\subsection{Path-like cactus}

Let $\{v_{1},\ldots,v_{n}\}$ be a path and $(G_{1},v_{1}),\ldots,(G_{n},v_{n})$ disjoined rooted graphs as is shown in Figure \ref{fig:GraphsOnPath}.
\begin{figure}[h]
   \centering
      \includegraphics[scale=1.1]{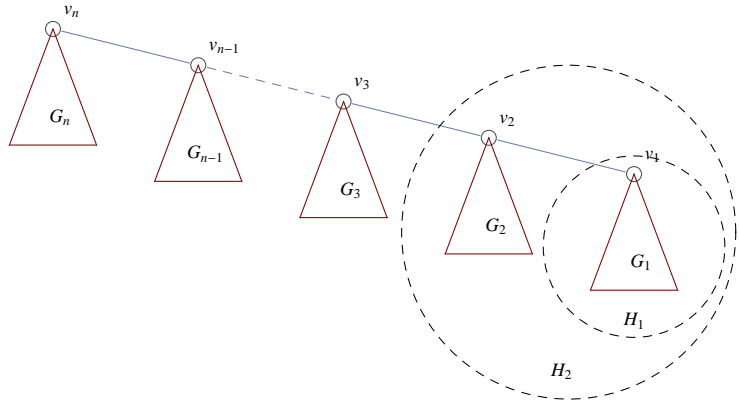}\\
      \caption{{\small Path-like cactus}}
      \label{fig:GraphsOnPath}
\end{figure}
Denote obtained graph by $G$ and consider it as a rooted graph $(G,v_{n})$.
\begin{lemma}
   Let $G_{1},G_{2},\ldots,G_{n}$ be disjoined graphs with specific vertices $v_{1},v_{2},\ldots$ $\ldots,v_{n}$ respectively and let $G$ be the disjoint union of $G_{1},G_{2},\ldots,G_{n}$, joined by the edges $v_{1}v_{2}, v_{2}v_{3},\ldots,v_{n-1}v_{n}$. For every $i \in \{1,\ldots,n\}$, denote by $H_{i}$ the union of graphs $G_{1},\ldots,G_{i}$, i.e. $ H_{i}=\left(\bigcup_{j=1}^{i}G_{j}\right) \cup \left(\bigcup_{j=1}^{i-1}(v_{j},v_{j+1})\right)$. If $i > 1$, the following is true
   \begin{eqnarray}
      \!\!\!\!\!\!\! \gamma_{w}^{00}(H_{i},v_{i}) & \! = &\! \gamma_{w}^{00}(G_{i},v_{i})+\gamma_{w}(H_{i-1}) \label{eq:pathlike:00} \\
      \!\!\!\!\!\!\! \gamma_{w}^{1}(H_{i},v_{i}) & \! = &\! \gamma_{w}^{1}(G_{i},v_{i})+\min\big\{ \gamma_{w}^{1}(H_{i-1},v_{i-1}),\gamma_{w}^{00}(H_{i-1},v_{i-1})\big\} \label{eq:pathlike:1} \\
      \!\!\!\!\!\!\! \gamma_{w}^{0}(H_{i},v_{i}) & \! = &\! \min\big\{\gamma_{w}^{0}(G_{i},v_{i})+\gamma_{w}(H_{i-1}), \gamma_{w}^{00}(G_{i},v_{i})+\gamma_{w}^{1}(H_{i-1},v_{i-1})\big\} \label{eq:pathlike:0} \\
      \!\!\!\!\!\!\! \gamma_{w}(H_{i}) & \! = &\! \min\big\{ \gamma_{w}^{1}(H_{i},v_{i}), \gamma_{w}^{0}(H_{i},v_{i}) \label{eq:pathlike:WDN} \big\} \,.
   \end{eqnarray}
   \label{lemma:pathlike}
\end{lemma}
\begin{proof}
   Look at the graph $H_{i}$ as the disjoint union of subgraphs $G_{i}$ and $H_{i-1}$ with roots $v_{i}$ and $v_{i-1}$ respectively and joined by the edge $v_{i-1}v_{i}$. These are exactly the assumptions of Lemma \ref{lemma:WDNofGRAPH2vertices}.
\end{proof}
\begin{algorithm}[H]
 \caption{PATH-LIKE}
 \begin{algorithmic}
 \State
 {\bf Data:} a path-like cactus $(P,r)$ with the DFS ordered path's vertices and corresponding parameters $\gamma_{w}^{00}$, $\gamma_{w}^{1}$, $\gamma_{w}^{0}$ and $\gamma_{w}$ (i.e. WDP and WDN of rooted subgraphs $(G_{i},v_{i})$ as is shown in Figure \ref{fig:GraphsOnPath}) \\
   {\bf set} $v$ is the last vertex in the DFS order; \\
   {\bf repeat}
      \begin{algorithmic}
      \State $w=\father(v)$; \\
      $\gamma_{w}^{00}(w) = \gamma_{w}^{00}(w)+\gamma_{w}(v)$; \\
      $\gamma_{w}^{1}(w) = \gamma_{w}^{1}(w) + \min\{\gamma_{w}^{1}(v),\gamma_{w}^{00}(v)\}$; \\
      $\gamma_{w}^{0}(w) = \min\{ \gamma_{w}^{0}(w)+\gamma_{w}(v), \gamma_{w}^{00}(w)+\gamma_{w}^{1}(v)\}$; \\
      $\gamma_{w}(w) = \min\{\gamma_{w}^{1}(w), \gamma_{w}^{0}(w) \}$; \\
      $v=w$;
      \end{algorithmic}
   {\bf until} {$v=r$.} \\
   {\bf Result:} $\gamma_{w}^{*}(P,r)=\gamma_{w}^{*}(v)$ \, for \, $*=00,1,0$;
   \begin{algorithmic}
           \State \qquad\quad $\gamma_{w}(P) = \gamma_{w}(v)$.
   \end{algorithmic}
   \end{algorithmic}
 \label{algorithm:PathLikeCactus}
\end{algorithm}
\begin{proposition}[(Time complexity of PATH-LIKE)]
   Algorithm PATH-LIKE needs $4(n-1)$ additions and $3(n-1)$ $\min$-operations.
\end{proposition}
\begin{proof}
   By counting all operations in (\ref{eq:pathlike:00}), (\ref{eq:pathlike:1}), (\ref{eq:pathlike:0}) and (\ref{eq:pathlike:WDN}), the proposition follows.
\end{proof}

\subsection{$D$-closed path-like cactus}

Let $\{v_{1},\ldots,v_{n}\}$ be a path and $(G_{1},v_{1}),\ldots,(G_{n},v_{n})$ disjoined rooted graphs. We require that both $v_{1}$ and $v_{n}$ are members of a dominating set. Such a graph $G$ is drawn on Figure \ref{fig:GraphsOnDclosedPath}.
\begin{figure}[h]
   \centering
      \includegraphics[scale=1.1]{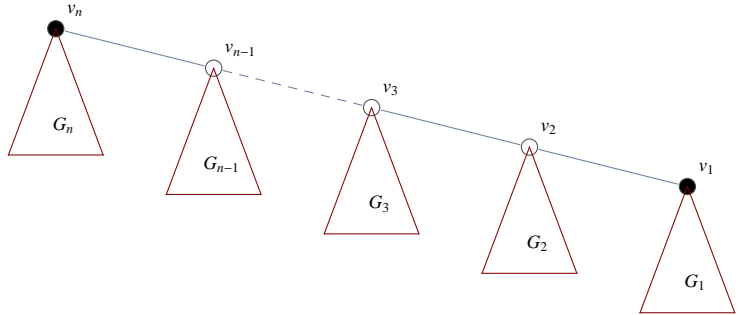}\\
      \caption{{\small D-closed path-like cactus}}
      \label{fig:GraphsOnDclosedPath}
\end{figure}
To calculate the weighted domination parameters and the weighted domination number with the condition that $v_{1} \in D$, we introduce some additional  notation:
(for $u \ne v_{1}$)
\begin{eqnarray*}
   \gamma_{w,v_{1}}(G) & = & \gamma_{w}^{1}(G,v_{1}) \\
    \gamma_{w,v_{1}}^{00}(G,u) & = & \gamma_{w,v_{1}}(G-u) \\
   \gamma_{w,v_{1}}^{1}(G,u) & = & \min\big\{ w(D) \,\big| \; \{u,v_{1}\} \subseteq D \big\} \\
   \gamma_{w,v_{1}}^{0}(G,u) & = & \min\big\{ w(D) \,\big| \; v_{1} \in D, u \notin D \big\} \,.
\end{eqnarray*}

In the new algorithm for calculating the WDN of a $D$-closed path-like cactus we have to provide that the first vertex $(v_{1})$ is a member of a dominating set.
We apply algorithm PATH-LIKE and make changes in the first step of the loop {\bf repeat-until}. According to the Figure \ref{fig:GraphsOnPath} and Figure \ref{fig:GraphsOnDclosedPath}, that means
\begin{eqnarray}
      \gamma_{w,v_{1}}^{00}(H_{2},v_{2}) & = & \gamma_{w}^{00}(G_{2},v_{2})+\gamma_{w}^{1}(H_{1},v_{1}) \\
      \gamma_{w,v_{1}}^{1}(H_{2},v_{2}) & = & \gamma_{w}^{1}(G_{2},v_{2})+\gamma_{w}^{1}(H_{1},v_{1}) \\
      \gamma_{w,v_{1}}^{0}(H_{2},v_{2}) & = & \gamma_{w}^{00}(G_{2},v_{2})+\gamma_{w}^{1}(H_{1},v_{1}) \,.
\end{eqnarray}
The other steps do not need corrections and the vertex $v_{1}$ on Figure \ref{fig:GraphsOnDclosedPath} (the last vertex in the DFS order in the algorithm below) on a path remains in a dominating set.

\begin{algorithm}[H]
 \caption{$D$-CLOSED PATH-LIKE}
 \begin{algorithmic}
 \State
 {\bf Data:} {$D$-closed path-like cactus $(P,r)$: with DFS ordered path's vertices and corresponding parameters $\gamma_{w}^{00}$, $\gamma_{w}^{1}$,
 $\gamma_{w}^{0}$ and $\gamma_{w}$, i.e. WDP and WDN of rooted subgraphs $(G_{i},v_{i})$} \\
   {\bf set} $v$ is the last vertex in DFS order; \\
   $l=v$; \\
   $w=\father(v)$;\\
   $\gamma_{w}^{00}(w) = \gamma_{w}^{00}(w)+\gamma_{w}^{1}(v)$; \\
   $\gamma_{w}^{1}(w) = \gamma_{w}^{1}(w) + \gamma_{w}^{1}(v)$; \\
   $\gamma_{w}^{0}(w) = \gamma_{w}^{00}(w)+\gamma_{w}^{1}(v)$; \\
   $\gamma_{w}(w) = \min\{\gamma_{w}^{1}(w), \gamma_{w}^{0}(w) \}$; \\
   $v=w$;\\
   {\bf repeat}  \\
      \quad $w=\father(v)$; \\
      \quad $\gamma_{w}^{00}(w) = \gamma_{w}^{00}(w)+\gamma_{w}(v)$; \\
      \quad $\gamma_{w}^{1}(w) = \gamma_{w}^{1}(w) + \min\{\gamma_{w}^{1}(v),\gamma_{w}^{00}(v)\}$; \\
      \quad $\gamma_{w}^{0}(w) = \min\{ \gamma_{w}^{0}(w)+\gamma_{w}(v), \gamma_{w}^{00}(w)+\gamma_{w}^{1}(v)\}$; \\
      \quad $\gamma_{w}(w) = \min\{\gamma_{w}^{1}(w), \gamma_{w}^{0}(w) \}$; \\
      \quad $v=w$; \\
   {\bf until} {$v=r$.} \\
   {\bf Result:} ${\gamma}_{w,l}^{*}(P,r)=\gamma_{w}^{*}(v)$ \, for \, $*=00,1,0$; \\
            \qquad \qquad ${\gamma}_{w,l}(P) = \gamma_{w}(v)$.
 \end{algorithmic}
 \label{algorithm:DclosedPathLikeCactus}
\end{algorithm}

\begin{proposition}[(Time complexity of $D$-CLOSED PATH-LIKE)] \,
   Al\-gorithm $D$-CLOSED PATH-LIKE needs less than $4(n-1)$ additions and $3(n-1)$ $\min$-operations.
\end{proposition}
\begin{proof}
   In the first step of the algorithm we have $3$ additions and one $\min$-operation. For the loop we need $4(n-2)$ additions and $3(n-2)$ $\min$-operations.
\end{proof}

\subsection{Cycle-like cactus}

We now consider the case when the specific vertices $v_{1},\ldots,v_{n}$ in a graph are vertices of a cycle.
 Let $(C_{n},v_{n})$ be a rooted cycle with vertices $v_{1},\ldots,v_{n}$ and corresponding weights $w_{1},\ldots,w_{n}$, and let $(G_{1},v_{1}),\ldots,(G_{n-1},v_{n-1})$ be disjoined rooted graphs (in our case cacti).
Denote by  $(K_{n},v_{n})$  the union of $(G_{1},v_{1}),\ldots$ $\ldots,(G_{n-1},v_{n-1})$ and $v_{n}$,
  joined by the edges $v_{1}v_{2},v_{2}v_{3},\ldots,v_{n-1}v_{n}$ and $v_{n}v_{1}$.
Graph $(K_{n},v_{n})$ is depicted  in Figure \ref{fig:GraphsOnCycle}.

\begin{figure}[h]
   \centering
      \includegraphics[scale=0.6]{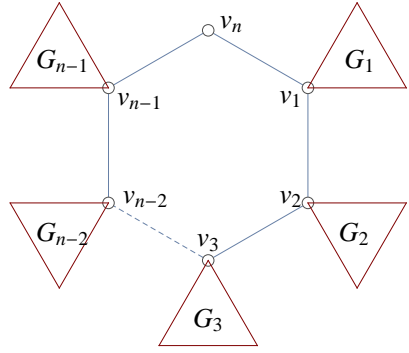}\\
      \caption{{\small Cycle-like cactus}}
      \label{fig:GraphsOnCycle}
\end{figure}
Let $(K_{n}',v_{n})$ be the path-like cactus with specific vertices on the weighted path $\{v_{n}',v_{1},\ldots,v_{n}\}$, where we additionally define $w(v_{n}')=w_{n}$. Graph $(K_{n}',v_{n})$ is obtained from $(K_{n},v_{n})$ as is shown  in Figure \ref{fig:GraphsOnCycleDrugic}.
\begin{figure}[h]
   \centering
      \includegraphics[scale=1.4]{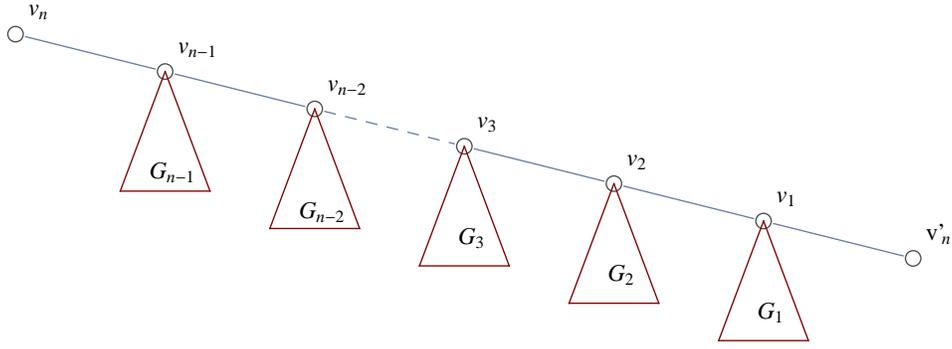}\\
      \caption{{\small Graph $(K_{n}',v_{n})$ obtained from $(K_{n},v_{n})$}}
      \label{fig:GraphsOnCycleDrugic}
\end{figure}

\begin{lemma}
   For a rooted cycle-like cactus $(K_{n},v_{n})$, the weighted domination parameters and the weighted domination number are the following
   \begin{eqnarray}
      \gamma_{w}^{00}(K_{n},v_{n}) & = & \gamma_{w}(K_{n}-v_{n}) \label{eq:cyclelike:00} \\
      \gamma_{w}^{1}(K_{n},v_{n}) & = & \gamma_{w,v_{n}'}^{1}(K_{n}',v_{n})-w_{n} \label{eq:cyclelike:1} \\
      \gamma_{w}^{0}(K_{n},v_{n}) & = & \min\big\{ \gamma_{w}^{1}(K_{n}-v_{n},v_{1}),\gamma_{w}^{1}(K_{n}-v_{n},v_{n-1}) \big\} \label{eq:cyclelike:0} \\
      \gamma_{w}(K_{n}) & = & \min\big\{ \gamma_{w}^{1}(K_{n},v_{n}), \gamma_{w}^{0}(K_{n},v_{n}) \big\} \label{eq:cyclelike:WDN} \,.
   \end{eqnarray}
\end{lemma}

\begin{proof}
 Lemma is a direct  consequence of the construction of the graph $K'$, the definition of $\gamma_{w,v_{n}'}^{1}$ and the properties of the parameters $\gamma_{w}^{0}$ and $\gamma_{w}^{00}$.

Recall that  weighted domination parameters $\gamma_{w}^{00}(K_{n},v_{n})$, $\gamma_{w}^{1}(K_{n},v_{n})$ and $\gamma_{w}^{0}(K_{n},v_{n})$ can be calculated using the previous two lemmas.
\end{proof}
\begin{algorithm}[H]
   \caption{CYCLE-LIKE}
   \begin{algorithmic}
   \State
   {\bf Data:} {a rooted cycle $(C,r)$: with DFS ordered vertices and corresponding parameters $\gamma_{w}^{00}$, $\gamma_{w}^{1}$, $\gamma_{w}^{0}$ and $\gamma_{w}$, i.e. WDP and WDN of rooted subgraphs $(G_{i},v_{i})$} \\
   {\bf set} \, $v$ is the last vertex in DFS order; \\
    \quad \; \; $r$ is the first vertex in DFS order; \\
    \quad \; \; $s=\successor(v)$; \\
    \quad \; \; $l=v$; \\
    {\bf set} new vertex $r'$ in the DFS table with:  \\
    \qquad \ $\dfn(r')=\dfn(v)+1$; \\
    \qquad \ $\father(r')=v$; \\
    \qquad \ ($\mmark(r')=1$, $\rooot(r')=r'$, $\successor(r')=r'$, $\ind(r')=1$); \\
    \qquad \ $\gamma_{w}^{00}(r')=\gamma_{w}^{00}(r)$, \ $\gamma_{w}^{1}(r')=\gamma_{w}^{1}(r)$, \ $\gamma_{w}^{0}(r')=\gamma_{w}^{0}(r)$ \ and \ $\gamma_{w}(r')=\gamma_{w}(r)$. \\
    {{\bf calculate}
     \begin{enumerate}
      \item [$\bullet$] $\gamma_{w}(C-r)$ and $\gamma_{w}^{1}(C-r,s)$ using PATH-LIKE on DFS ordered path's vertices $\{s,\ldots,v\}$
      \item [$\bullet$] ${\gamma}_{w,r'}^{1}(C \cup \{r'\},r)$ using D-CLOSED PATH-LIKE on DFS ordered path's vertices $\{r,s,\ldots,v,r'\}$
      \item [$\bullet$] ${\gamma}_{w,v}(C-r)$ using D-CLOSED PATH-LIKE ($v \in D$) on DFS ordered path's vertices $\{s,\ldots,v\}$
      \end{enumerate}
    }
    {\bf Result:}  $\gamma_{w}^{00}(C,r)=\gamma_{w}(C-r)$; \\
     \qquad \qquad \quad \; \; $\gamma_{w}^{1}(C,r)={\gamma}_{w,r'}^{1}(C \cup \{r'\},r)-\gamma_{w}^{1}(r)$; \\
     \qquad \qquad \quad \; \; $\gamma_{w}^{0}(C,r)=\min\big\{\gamma_{w}^{1}(C-r,s),{\gamma}_{w,v}(C-r)\big\}$; \\
     \qquad \qquad \quad \; \; $\gamma_{w}(C)=\min\big\{ \gamma_{w}^{1}(C,r),\gamma_{w}^{0}(C,r)\big\}$.
 \end{algorithmic}
 \label{algorithm:Cycle-likeCactus}
\end{algorithm}

\begin{proposition}[(Time complexity of CYCLE-LIKE)] $ $ \\
   If a cycle $C$ has $n$ vertices,
   the algorithm CYCLE-LIKE needs less than $12(n-1)$ additions and $9(n-1)$ $\min$-operations.
\end{proposition}

\begin{proof} Using algorithms {PATH-LIKE} and {D-CLOSED PATH-LIKE} we obtain:
   \begin{enumerate}
      \item [$\bullet$] For calculating the parameters $\gamma_{w}(C-r)$ and $\gamma_{w}^{1}(C-r,s)$ using {PATH-LIKE} algorithm on $n-1$ vertices, we need $4(n-2)$ additions and $3(n-2)$ $\min$-operations.
      \item [$\bullet$] For calculating the parameter $\overline{\gamma}_{w}^{1}(C \cup \{r'\},r)$ using D-CLOSED PATH-LIKE algorithm on $n+1$ vertices, we need $4n$ additions and $3n$ $\min$-operations.
      \item [$\bullet$] For calculating $\overline{\gamma}_{w}(C-r)$ using D-CLOSED PATH-LIKE algorithm on $n-1$ vertices, we need $4(n-2)$ additions and $3(n-2)$ $\min$-operations.
   \end{enumerate}
   Additionally, at the end of the algorithm, we need one addition and two $\min$-operations. Adding up all operations, we confirm
   \begin{equation}
      4(n-2)+4n+4(n-2)+1 = 12n-15 < 2(n-1)
   \end{equation}
   additions and
   \begin{equation}
      3(n-2)+3n+3(n-2)+2 = 9n-10 < 9(n-1)
   \end{equation}
   $\min$-operations.
\end{proof}

\section{Algorithm for weighted domination of cacti}

\label{section:summary}

As we indicated  in the previous sections, the general  algorithm for calculating   WDN of a cactus graph can be
seen as an   upgrade of the algorithm TREE. The input data of the main algorithm is the DFS cactus  data structure, which is supplemented by four arrays of the initial values of the parameters $\gamma_{w}^{00}$, $\gamma_{w}^{1}$, $\gamma_{w}^{0}$ and $\gamma_{w}$ for every vertex.
The starting point (vertex) of the algorithm is the last unread vertex in the DFS order. If the last vertex lies on a tree, we proceed  like in the algorithm TREE. (Some care must be taken for the root of the tree, to correct its parameters following Lemma \ref{lemma:WDNofGRAPH1vertex}.)
However, if the last vertex lies on a cycle, we have to read and remember all cycle's vertices.
Following Remark \ref{remark:subcactus} and Proposition \ref{prop:last vertex cycle}, the algorithm calls itself for rooted subcacti,  for which the roots are hinges of the cycle.
This forces the hinges to obtain new values of parameters where all subcacti rooted at the hinges are considered.
 Then the algorithm calls subalgorithm CYCLE-LIKE and applies Lemma \ref{lemma:WDNofGRAPH1vertex} to correct the values of the parameters of the root of cycle.
The algorithm continues until the last unread vertex in the DFS order is the root of the cactus.
Pseudocode is given below (Algorithm \ref{AlgCactus}). \\

\begin{algorithm}
   \caption{CACTUS}
   \begin{algorithmic}
   \State
   {\bf Data:} A rooted cactus $(K,r)$ with DFS ordered vertices in the DFS table; \\
   {\bf initialize} $\gamma_{w}^{00}(v)=0$, $\gamma_{w}^{1}(v)=w(v)$, $\gamma_{w}^{0}(v)=\infty$ and $\gamma_{w}(v)=w(v)$ for 
                  every vertex $v$ in the DFS table; \\
 {\bf set} $v$ is the last vertex in the DFS order \\
    \qquad $\last = v$; \\
 {\bf while} $\last \ne r$ {\bf do}
       \begin{algorithmic}
       \If{ $\last$ {\rm does not lie on a cycle}} \\
            \quad {\bf repeat} \\
            \quad \,  $w = \father(v)$; \\
            \quad \,  $u = v$;
             \If{($\rooot(w)=w$) and ($\dfn(w) < \dfn(v)-1$)} \\
                        \qquad \quad {{\bf do} CACTUS of the rooted subcactus on vertices in the DFS table \\
                        \qquad \quad with $\dfn=\dfn(w),\ldots,\dfn(v)-1$ and the root $w$ \\
                        \qquad \quad {\it (we get new values $\gamma_{w}^{00}(w)$, $\gamma_{w}^{1}(w)$, $\gamma_{w}^{0}(w)$, $\gamma_{w}(w)$)}
                         }
            \EndIf \\
            \quad \, $\gamma_{w}^{00}(w)=\gamma_{w}^{00}(w)+\gamma_{w}(v)$ ; \\
            \quad \, $\gamma_{w}^{1}(w)=\gamma_{w}^{1}(w)+\min\{\gamma_{w}^{1}(v), \gamma_{w}^{00}(v) \}$ ; \\
            \quad \, $\gamma_{w}^{0}(w)= \min\{\gamma_{w}^{0}(w)+\gamma_{w}(v), \gamma_{w}^{00}(w)+\gamma_{w}^{1}(v) \} $; \\
            \quad \, $\gamma_{w}(w)=\min\{\gamma_{w}^{1}(w),\gamma_{w}^{0}(w) \}$; \\
            \quad \,  $v = w$; \\
            \quad {\bf until} ($\rooot(v)\ne v$) or ($\dfn(v)=0$)
         \Else (see next page)
      \EndIf
      \end{algorithmic}
  \end{algorithmic}
\end{algorithm}

\begin{algorithm}
   \caption{CACTUS - Part 2}
   \begin{algorithmic}
      \State
      \begin{algorithmic}
       \If{$\last$ {\rm does not lie on a cycle}} see previous page
       \Else \\
       \quad $v = \last$ ; \\
       \quad {\bf repeat} {\it (mark the roots of the cycle and correct WDP and WDN arrays \\
                           \qquad \qquad \;\, of hinges on a cycle)} \\
            \quad \,  $w=\father(v)$;
                     \If{$\dfn(w) < \dfn(v)-1$} \\
                         \qquad \quad {{\bf do} CACTUS of rooted subcactus on vertices in the DFS table with \\
                         \qquad \quad $\dfn=\dfn(w),\ldots,\dfn(v)-1$ and the root $w$ \\
                         \qquad \quad{\it(new values $\gamma_{w}^{*}(w)$, $*=00,1,0$ and $\gamma_{w}(w)$)}}
                     \EndIf \\
            \quad \,   $v = w$; \\
       \quad {\bf until} $v=\successor(v)$. \\
       \quad $ u = v $; \\
       \quad $ w= \father(v) $; \ \  {\it ($w$ is now the root of the cycle)} \\
       \quad {\bf Make} cycle table: \\
          \quad \ \ \ $C = \emptyset$; \\
          \quad \ \ \ $v=\last$; \\
          \quad \ \ \ $C = DFS(v)$; \\
          \quad \ \ \ {\bf repeat} \\
          \quad \ \ \ \ \ \ $v = \father(v)$; \\
          \quad \ \ \ \ \ \ $C \cup DFS(v)$; \\
          \quad \ \ \ {\bf until} $v=w$.
          \\
       \quad Do CYCLE-LIKE algorithm on the rooted cycle $(C,v)$ with vertices in \\
       \quad the table $C$. We obtain parameters of the cycle $C$: $\gamma_{w}^{00}(C,v)$, $\gamma_{w}^{1}(C,v)$, \\
       \quad $\gamma_{w}^{0}(C,v)$ and $\gamma_{w}(C)$; \\
       \quad $\gamma_{w}^{00}(w)=\gamma_{w}^{00}(w)+\gamma_{w}^{00}(C,v)$; \\
       \quad $\gamma_{w}^{1}(w)=\gamma_{w}^{1}(w)+\gamma_{w}^{1}(C,v)-\weight(w)$;\\
       \quad $\gamma_{w}^{0}(w)=\min\{\gamma_{w}^{0}(w)+\gamma_{w}^{00}(C,v),\gamma_{w}^{00}(w)+\gamma_{w}^{0}(C,v)\}$ ; \\
       \quad $\gamma_{w}(w)=\min\{\gamma_{w}^{1}(w), \gamma_{w}^{0}(w) \}$;
       \EndIf
       \end{algorithmic}
       \quad$\last$ is determined by $\dfn(\last)=\dfn(u)-1$. \\
   {\bf Result:} ${\gamma}_{w}^{*}(K,r)=\gamma_{w}^{*}(w)$ \, for \, $*=00,1,0$; \\
                 \qquad \qquad ${\gamma}_{w}(K) = \gamma_{w}(w)$.
   \end{algorithmic}
\label{AlgCactus}
\end{algorithm}


\newpage



Denote by $b$ the number of blocks i.e. the total number of cycles and grafts.

\begin{proposition}
   Algorithm CACTUS properly calculates the weighted domination number of a cactus.
\end{proposition}

\begin{proof}
   Recall that by definition we have two essential situations. If the current vertex $v$ is a $G$-vertex on a subtree or a root of a cycle, the algorithm CACTUS calculates the WDP and the WDN of the subcactus of all vertices with $\dfn \geq \dfn(v)$. In particular, when $v=r$ the algorithm CACTUS calculates the WDP and the WDN
   of the given cactus and we have
   \begin{eqnarray}
      \gamma_{w}(K)=\gamma_{w}(r)\,.
   \end{eqnarray}
\end{proof}
\medskip

Below we  show that algorithm CACTUS needs less than $12n+5b$ additions and $9n+2b$ $\min$ operations and thus prove Theorem 1.

\begin{proof}  {\bf (of Theorem 1.)}
   Let $B_{1},\ldots,B_{b}$ be blocks in the cactus and denote by $n_{j}$ the number of vertices in the block $B_{j}$ for each $j=1,\ldots,b$. Since a hinge can be the root of more than one block, the number of hinges is less or equal $b$. Therefore, we have the inequality
   \begin{equation}
      n_{1} + n_{2} + \ldots + n_{b} -b \leq n \,.
      \label{ineq:vertices and hinges}
   \end{equation}
   Since the algorithm requires much more  time for a cycle block (in comparison with a graft),
we can estiamate
 that for each block $B_{j}$ we need less than   $12(n_{j}-1)$ additions and $9(n_{j}-1)$ $\min$-operations.
Furthermore, 5 additions and 2 $\min$-operations are needed for sticking blocks in a hinge.
Summing up all operations for all blocks in the cactus, we get
   \begin{equation}
      \sum_{j=1}^{b} 12(n_{j}-1) + 5b = 12(\sum_{j=1}^{b} n_{j} - b ) + 5b \leq 12n + 5b \nonumber
   \end{equation}
   and
   \begin{equation}
      \sum_{j=1}^{b} 9(n_{j}-1) + 2b \leq 9n+2b \,.   \nonumber
   \end{equation}
\end{proof}

\begin{rem}
In the proof above, we have assumed that the DFS data structure of the cactus is given.
The reason is that the algorithm for $k$-trees \cite{PAPkT} assumes the structural information of a partial $k$-tree is given.
It is however well-known that the DFS algorithm is linear in the number of edges of a graph, which for trees and cactus graphs implies that it is also linear in the number of vertices.
More precisely,  $4m$ operations are needed when traversing the graph during  DFS
that provides the DFS cactus data structure: the DFS search has to be followed by a traversal in the opposite DFS order and, in addition. each cycle has to be traversed two more times to
assign the roots and  the successors to all vertices of a cycle.
\end{rem}

\noindent
{\bf Acknowledgement} \\

This work was supported in part by Slovenian Research Agency ARRS (Grant number P1-0285-0101).

\bigskip

\bibliographystyle{plainnat}


\end{document}